\documentclass[letterpaper, 10 pt, conference]{cssconf}

\IEEEoverridecommandlockouts                              

\usepackage{amsmath, amsfonts}
\usepackage{graphicx}
\usepackage{color}
\usepackage{colortbl}

\newtheorem{corollary}{Corollary}[section]
\newtheorem{lemma}{Lemma}[section]
\newtheorem{theorem}{Theorem}[section]

\def\theequation{\thesection.\arabic{equation}}

\newcommand{\be}{\begin{equation}}
\newcommand{\ee}{\end{equation}}
\newcommand{\bd}{\begin{displaymath}}
\newcommand{\ed}{\end{displaymath}}
\newcommand{\bea}{\begin{eqnarray}}
\newcommand{\eea}{\end{eqnarray}}

\newcommand{\R}{\mathbb{R}}

\newcommand{\Z}{\mathbb{Z}}

\title{\LARGE \bf On the Zeros of Quasi-Polynomials with Single Delay}
\author{Suat Gumussoy\thanks{Courant Institute of Mathematical Sciences, New York University}}
\date{}

\begin{document}
\maketitle

\begin{abstract}
A new numerical method is introduced for calculation of
quasi-polynomial zeros with constant single delay. The trajectories
of zeros are obtained depending on time-delay from zero to final
time-delay value. The method determines all the zeros of the
quasi-polynomial in any right half-plane. The approach is used to
determine stability analysis of time-delay systems. The method is
easy to implement, robust and applicable to quasi-polynomials with
high order. The effectiveness of the method is shown on an
example.\end{abstract}

\section{Introduction}
The stability of time-delay system is determined by zeros of
quasi-polynomial. The linear time-invariant time-delay system is
stable if the zeros of the characteristic equation are on the left
half-plane. There are many methods in the literature checking the
stability of time-delay systems without computing zeros of the
characteristic equation using Lyapunov theory
\cite{LyapMethodIvanescu, LyapMethodYong, LyapMethodLi}, matrix
measures \cite{MatrixMeasureArun, MatrixMeasureHu}, matrix pencil
techniques \cite{MatrixPencilNiculescu}, direct methods
\cite{DirectMethodSilva, DirectMethodOlgac} (for general survey, see
\cite{NiculescuDelayEffects, Richard2003, GuBook2003} and the
references therein).

An alternative approach for analysis of time-delay systems is to
locate and compute all the zeros of the quasi-polynomial (i.e.,
characteristic equation) in the region $\textrm{Re(s)}>\sigma_0$.
This approach has two advantages over the existing analytical and
graphical tests:
\begin{itemize}
  \item the response of the system can be approximately determined by zeros of the characteristic equation,
  \item the location of the zeros of the quasi-polynomial gives more information than asymptotic stability
  such as $\cal{D}$-stability, fragility of the system, stability
  margins, robustness.
\end{itemize}

The rightmost characteristic roots can be computed by discretization
of infinitesimal generator or the time integration operator of delay
differential equation \cite{EngelBorghs2000}. The infinitesimal
generator of the solution operators semigroup is approximated by
Runge-Kutta \cite{Breda2004} and linear multistep method
\cite{Breda2005, Breda2006}. The rightmost characteristic roots are
approximated by computing dominant eigenvalues of the discretized
time integration operator. The
 discretized version of the delay differential equation is
represented by linear multistep method for next time-delay step
calculation \cite{EngelBorghs2002, Verheyden2007}.

In this paper, we present complex-valued differential equation to
trace the trajectory of the zero of the quasi-polynomial as a
function of time-delay. The numerical solution of differential
equation is obtained by $4^\textrm{th}$ order Runge-Kutta method
with variable step size. All the zeros of the quasi-polynomial can
be traced in the predefined right half-plane using our method. Since
the zeros and their trajectories are available by the method, it is
an important tool for comprehensive stability analysis of time-delay
systems.

Our method finds all the zeros of the quasi-polynomial with single
delay in user-defined right half-plane. Since the location of the
zeros are known unlike stability checking methods
\cite{Richard2003}, the stability of the system can be analyzed in
detail (i.e., fragility of the system, real stability abscissa,
$\mathcal{D}$-stability, robustness, stability margins). Although
there are other methods in the literature calculating the rightmost
roots of the quasi-polynomials \cite{Breda2004, Breda2005,
Breda2006}, \cite{EngelBorghs2000} \cite{EngelBorghs2002},
\cite{Verheyden2007}, these methods approximate the zeros of the
quasi-polynomial for a given time-delay. Our method \emph{traces}
the zeros of quasi-polynomial from $0$ to final time-delay value in
the user-defined region, therefore it is not approximation. In
practice, the stability of time-delay system is not analyzed for a
particular time-delay value but an interval due to variation in the
delay. Since our method finds the location of all the zeros of
quasi-polynomial from $0$ to final delay value in the region, it is
effective to analyze the \emph{delay effect} on system stability.

The rest of the paper is organized as follows. The quasi-polynomial
definition and characteristics of its zeros are described in Section
\ref{sec:qpolychar}. The numerical method for trajectories of zeros
depending time-delay and its analysis are given in Section
\ref{sec:method}. The method is applied to an example for stability
analysis of time-delay system in Section \ref{sec:example}.
Concluding remarks are provided in Section \ref{sec:concl}.

\section{Quasi-Polynomial Definition and Characteristics of its Zeros} \label{sec:qpolychar}
Consider the \textit{retarded type} quasi-polynomial, \be
\label{eq:qpoly} f(s,h)=a(s)+b(s)e^{-hs} \ee where $a$, $b$ are
polynomials of real coefficients with degrees $n_a$, $n_b$ and $h$
is a positive real number. Since the quasi-polynomial is retarded
type, the degree of $a$ is larger than the degree of $b$, i.e.,
$n_a>n_b$.

It is known that the zeros of the quasi-polynomial (\ref{eq:qpoly})
are infinitely many in the left half-plane and finitely many in the
right half-plane for fixed time-delay value. The infinitely many
zeros extend to minus infinity with decreasing real and imaginary
part.

\begin{figure}
\centering \includegraphics[width=.4\textwidth]{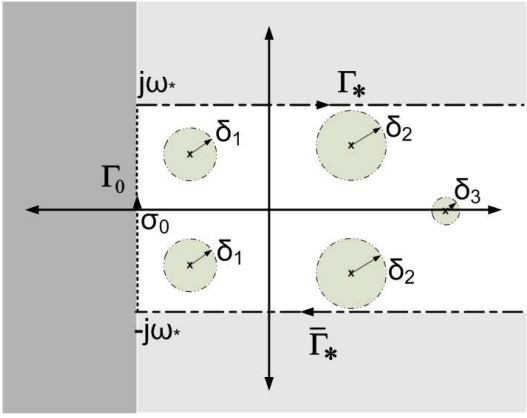}
\caption{Regions of Zeros for $f(s)$} \label{fig:ZerosRegion}
\end{figure}

The following Lemma establishes the connection
between the zeros of quasi-polynomial when $h=0$ and $h>0$.
\begin{lemma} \label{lemma:zeroregion}
For any $\sigma_0\in\R$ which is smaller than minimum real part of
roots of $a(s)+b(s)$, there exists $h>0$ such that the number of the
zeros of the quasi-polynomial (\ref{eq:qpoly}) are equal to the
number of the zeros of $a(s)+b(s)$ in the region
$\textrm{Re(s)}>\sigma_0$. All the infinitely many zeros of the
quasi-polynomial (\ref{eq:qpoly}) are in the complementary region,
i.e., $\textrm{Re(s)}<\sigma_0$.
\end{lemma}
\begin{proof}
Fix $\sigma_0$ as defined in the Lemma, since the degree of
polynomial $a$ is larger than polynomial $b$, it is possible to
choose $\omega_*$ large enough such that \bd
\sup_{s\in\Gamma_*}\left|\frac{b(s)}{a(s)+b(s)}\right|<\epsilon_*
\ed holds where $\Gamma_*=\{s\in\sigma+j\omega_*$ for
$\sigma_0<\sigma<\infty\}$. Calculate \bd
r_0\doteq\max_{s\in\Gamma_0}\left|\frac{b(s)}{a(s)+b(s)}\right| \ed
where $\Gamma_0=\{s\in\sigma_0+j\omega$ for
$-\omega_*<\omega<\omega_*\}$. The real number $r_0$ is well-defined
 since the curve $\Gamma_0$ passes from left of the zeros of the polynomial $a+b$. Choose $h$ small enough such that
\bd
\max_{s\in\Gamma_0}|e^{-hs}-1|<r_0^{-1}
\ed is satisfied.

The functions $a(s)+b(s)$ and $b(s)(e^{-hs}-1)$ are analytic on and
inside the contour $\Gamma=\Gamma_0\bigcup \Gamma_*\bigcup
\bar{\Gamma}_*$. Note that on the contour $\Gamma$, \bd
\max_{s\in\Gamma}\left|\frac{b(s)(e^{-hs}-1)}{a(s)+b(s)}\right|<1
\ed holds. Using \textit{Rouch\'{e}'s Theorem} (see
\cite{Conway1978}, Theorem 3.8) , the functions $a(s)+b(s)$ and
$a(s)+b(s)e^{-hs}$ have the same number of zeros inside $\Gamma$.
The light gray region of Figure \ref{fig:ZerosRegion} satisfies, \bd
|a+be^{-hs}|>|b|(\epsilon_*^{-1}-1-e^{-h\sigma_0})>0 \ed for
sufficiently small $\epsilon_*$. Therefore, all other zeros lie in
$\textrm{Re(s)}<\sigma_0$ which is dark gray region in Figure
\ref{fig:ZerosRegion}.
\end{proof}

The Lemma \ref{lemma:zeroregion} shows that for sufficiently small
$h$, all the zeros of the quasi-polynomial (\ref{eq:qpoly}) are
close to the zeros of $a+b$ and all the other zeros can be pushed to
the left by choosing smaller $h$. Therefore, one can choose the
region large or small depending on requirements of stability
analysis. This qualitative behavior makes boundary crossing
detection possible.

The following Lemma can be used to calculate the distance of the
zero of $f(s,h)$ from zero of $f(s,h+\Delta h)$.

\begin{lemma} \label{lemma:distancezero}
Let $s_i^h$  be the zero of $f(s,h)$ (\ref{eq:qpoly}), and
$s_i^{h+\Delta h}$ be the zero of $f(s,h+\Delta h)$ where $s_i^h$
and $s_i^{h+\Delta h}$ are on the same trajectory. The distance
between $s_i^h$ and $s_i^{h+\Delta h}$, i.e., $\delta_i=|\Delta
s_i|=|s_i^{h+\Delta h}-s_i^h|$ is, \bea \nonumber
\delta_i&\leq&\frac{|s_i^h| |\Delta
h|}{\left|\frac{a'(s_i)}{a(s_i)}-\frac{b'(s_i)}{b(s_i)}+h
\right|}+\mathcal{O}(|\Delta h|^2) \eea where $\mathcal{O}()$ is the
order function.
\end{lemma}
\begin{proof}
 Taylor expansion for $f(s,h)$ can be written as,
\begin{multline}
f(s_i^{h+\Delta h},h+\Delta h)=f(s_i^h,h)+f_s(s_i^h,h)\Delta
s_i\\+f_h(s_i^h,h)\Delta h+\mathcal{O}(|\Delta h|^2)
\end{multline} where $f_s$, $f_h$ are derivative of function $f(s,h)$ with
respect to $s$ and $h$. Since $s_i^{h+\Delta h}$ and $s_i^h$ are
zeros of $f(s,h)$, \bea \label{eq:discODE}
\nonumber 0&=&f_s(s_i^h,h)\Delta s_i+f_h(s_i^h,h)\Delta h+\mathcal{O}(|\Delta h|^2), \\
\Delta s_i&=&-\frac{f_h(s_i^h,h)}{f_s(s_i^h,h)}\Delta h+\mathcal{O}(|\Delta h|^2), \\
\nonumber \delta_i&\leq&\left|\frac{f_h(s_i^h,h)}{f_s(s_i^h,h)}\right||\Delta h|+\mathcal{O}(|\Delta h|^2), \\
\nonumber \delta_i&\leq&\left|\frac{-bs_i^he^{-hs_i^h}}{a'(s)+(b'(s)-hb(s))e^{-hs_i^h}}\right||\Delta h|+\mathcal{O}(|\Delta h|^2), \\
\nonumber
\delta_i&\leq&\left|\frac{a(s_i^h)s_i^h}{a'(s_i^h)-(\frac{b'(s_i^h)}{b(s_i^h)}-h)a(s_i^h)}\right||\Delta
h|+\mathcal{O}(|\Delta h|^2). \eea
\end{proof}
The following section describes the method to find trajectories of
quasi-polynomial zeros depending time-delay $h$. The convergence
analysis of the method is given. All trajectories of the zeros
except the asymptotic zeros are traced.

\section{The Method and Analysis} \label{sec:method}
\subsection{Trajectory of Quasi-Polynomial Zero in the Region}
\begin{theorem} \label{theorem:trajODE}
Let $s_i^h$ be the zero of the quasi-polynomial (\ref{eq:qpoly}).
The trajectory of $s_i^h$ can be determined by the complex-valued
ordinary differential equation, \be \label{eq:contODE}
\frac{ds}{dh}=-\left(\frac{f_h(s,h)+f(s,h)}{f_s(s,h)}\right), \quad
s(h)=s_i^h \ee where $s$ is the trajectory of the zero and $s_i^h$
is the value when time-delay is equal to $h$. The trajectory is
locally convergent to trajectory of zero.
\end{theorem}
\begin{proof}
Define the local Lyapunov function
$V(s,h)=\frac{1}{2}\overline{f(s,h)}f(s,h)$. Note that $V(s,h)$ is
positive definite and assumes the value $0$ locally at the desired
equilibrium $f(s,h)=0$. \bea
\nonumber \frac{dV(s,h)}{dh}&=&-Re(\overline{f(s,h)}(f_s(s,h)\frac{ds}{dh}+f_h(s,h)), \\
\nonumber &=&-|f(s,h)|^2\leq0. \eea Note that when the trajectory is
close to local zero, it is convergent. When the trajectory coincides
with the trajectory of the zero of the quasi-polynomial
(\ref{eq:qpoly}), using Taylor expansion one can show that \be
\label{eq:Otube} f(s_i^{h+\Delta h},h+\Delta
h)\approx\mathcal{O}(|\Delta h|^2). \ee
\end{proof}
\begin{figure}
\centering \includegraphics[width=.4\textwidth]{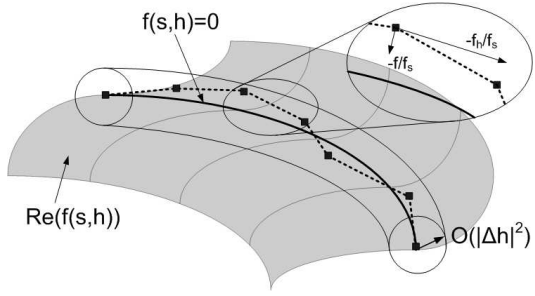}
\caption{Trajectory of $f(s,h)$ zero on the $\textrm{Re}(f(s,h))$
surface} \label{fig:MethodAnalysis}
\end{figure}

The exact solution of differential equation (\ref{eq:contODE})
follows the trajectory of the zero of the quasi-polynomial
(\ref{eq:qpoly}). Since the differential equation is solved by
numerical methods in general, the solution is not exact and
numerical errors are introduced due to step size of the numerical
algorithm. However, the differential equation is numerically robust
and the solution converges to trajectory of the zero of the
quasi-polynomial locally and the algorithm decreases its error
quadratically when step size is halved because of (\ref{eq:Otube}).

\subsection{Quasi-Polynomial Zeros Crossing the Boundary}
Given the region $\textrm{Re(s)}>\sigma_0$ for fixed
$\sigma_0\in\R$, it is possible to trace the trajectory of roots in
the region using (\ref{eq:contODE}) as time-delay $h$ changes. There
may be additional zeros crossing the boundary
$\textrm{Re(s)}=\sigma_0$ and entering into the region
$\textrm{Re(s)}>\sigma_0$.

It is clear that when $h=0$, the zeros of quasi-polynomial in the
region $\textrm{Re(s)}>\sigma_0$ are the roots of polynomial
$a(s)+b(s)$ inside the region. The method traces all these zeros
until they are outside the region. It is possible to find the
zero-crossing frequencies $\omega_i$ on the boundary
$s=\sigma_0+j\omega$ for $\omega\in\R$ using magnitude and phase
equations as \bea
\nonumber 0&=&a(\sigma_0+j\omega)+b(\sigma_0+j\omega)e^{-h(\sigma_0+j\omega)}, \\
\nonumber h(\omega)&=&-\frac{\ln{\left|\frac{a(\sigma_0+j\omega)}{b(\sigma_0+j\omega)}\right|}}{\sigma_0}, \\
\nonumber
0&=&h(\omega)\omega+2k\pi+\angle{-\frac{a(\sigma_0+j\omega)}{b(\sigma_0+j\omega)}},\;k\in\Z.
\eea The zero-crossing frequencies are $\omega_i\geq0$ values such
that the last equation is satisfied. Note that for each $\omega_i$,
the corresponding time-delay $h(\omega_i)$ should be positive. There
are finitely many $\omega_i$ satisfying the equation for fixed
time-delay value, $h$. When $\sigma_0=0$ is chosen, it is possible
to find analytical method to calculate $\omega_i$ and corresponding
$h(\omega_i)$ as explained in \cite{GuBook2003}.

The zeros of quasi-polynomial (\ref{eq:qpoly}),
$s=\sigma_0+j\omega_i$, cross the boundary at the time-delay values
$h(\omega_i)$ for $k=~1,\ldots,n_h$. However, we are interested in
the characteristic roots crossing the boundary and entering into the
region.

\begin{corollary} \label{cor:corRT}
The zero of quasi-polynomial (\ref{eq:qpoly}),
$s=~\sigma_0+j\omega_i$, crosses the boundary and enters into the
region $\textrm{Re(s)}>\sigma_0$  at the time-delay $h(\omega_i)$ if
\be \label{eq:RT} Re
\left[\frac{1}{s}\left(\frac{b'(s)}{b(s)}-\frac{a'(s)}{a(s)}-h(\omega_i)\right)
\right]_{\substack{s=\sigma_0+j\omega_i}}>0. \ee \vspace{.1cm}
\end{corollary}
\begin{proof}
The proof is based on \cite{Sipahi2006} and some algebraic
manipulations. One can show that \bea
\nonumber \frac{ds}{dh}&=&\frac{bse^{-hs}}{a'(s)+(b'(s)-h b(s))e^{-hs}}, \\
\nonumber &=&\frac{1}{\frac{1}{s}\left(\frac{a'(s)}{b(s)e^{-hs}}+\left(\frac{b'(s)}{b(s)}-h\right)\right)}, \\
\nonumber
&=&\frac{1}{\frac{1}{s}\left(\frac{b'(s)}{b(s)}-\frac{a'(s)}{a(s)}-h\right)}.
\eea The crossing zero enters the region if \bd
Re\left(\left.\frac{ds}{dh}\right|_{\substack{s=\sigma_0+j\omega_i
\\h=h(\omega_i)}
}\right)>0
\ed which is equivalent to the condition (\ref{eq:RT}).
\end{proof}

The Corollary \ref{cor:corRT} is very important. It reduces the
problem of the selection of the zeros entering into the region
$\textrm{Re(s)}>\sigma_0$ into inequality condition (\ref{eq:RT})
check for finite number of complex points, $s=j\omega_i$. Since we
are interested in the zeros entering the region, define the
zero-crossing points $s_{zc,k}=\sigma_0+j\omega_k$ satisfying
Corollary \ref{cor:corRT}. Note that the points
$s_{zc,k}=\sigma_0+j\omega_k$ cross the boundary for time-delays
$h(\omega_i)$ for $k=1,\ldots,n_{zc}$.

\subsection{The Method} \label{sec:themethod}
Our objective is to find all the zeros of quasi-polynomial
(\ref{eq:qpoly}) in the region $\textrm{Re(s)}>\sigma_0$ when
time-delay is equal to~$h$.
\begin{enumerate}
  \item Choose $\sigma_0\in\R$ according to design requirements,
  \item Define the maximum error tolerance, $\epsilon_{tz}$, of the trajectory of the zero,
  \item Calculate all the zero-crossing frequencies $\omega_i$ and the corresponding
  time-delay values $h(\omega_i)$ for \mbox{$i=1,\ldots,n_h$},
  \item Select the zero-crossing frequencies satisfying (\ref{eq:RT}) where the zeros enter
  into the region $\textrm{Re(s)}>\sigma_0$ (i.e., find $s_{sz,k}$ for $k=1,\ldots,n_{zc}$),
  \item Define the delay set $\mathbf{h}=\{0\}\bigcup_{i=1}^{n_{zc}}\{h(\omega_k)\}$ with elements in ascending order,
  \item Define the set of zeros in the region $\textrm{Re(s)}>\sigma_0$ as
  $\mathbf{z}=\{z\in\mathbb{C},\;a(z)+b(z)=0, \textrm{Re(z)}>\sigma_0\}$ for
  $\mathbf{h}_1=0$,
  \item Trace all the trajectories in the set $\mathbf{z}$ by
  (\ref{eq:contODE}) until next element of $\mathbf{h}$ (i.e., Runge-Kutta
  within the error tolerance $\epsilon_{tz}$). If any zero goes out
  of the region, drop the zero from the set $\mathbf{z}$. Add the new point
  \mbox{$s=\sigma_0+j\omega_i$}
  for next element of $\mathbf{h}$ into the set $\mathbf{z}$. Repeat
  this step for all the elements of $\mathbf{h}$,
  \item Trace the zeros in $\mathbf{z}$ until $h$. All the elements of $\mathbf{z}$ are the zeros of the
  quasi-polynomial (\ref{eq:qpoly}) in the region $\textrm{Re(s)}>\sigma_0$.
  \end{enumerate}

\textbf{Remarks:}
\begin{enumerate}
  \item The selection of $\sigma_0$ depends on objective of
  analysis. It can be selected close to imaginary axis to check the stability
  of the system. Since $\sigma_0$ can be chosen as
  desired, besides stability of the system, other stability criteria can be analyzed using our
  method such as damping criteria, $\mathcal{D}$-stability,
  fragility of the system, real stability distance, stability
  margins. It is possible to set $\sigma_0$ far from imaginary axis in the left half-plane such that the
  impulse response of the system is determined by zeros of the
  quasi-polynomial.
  \item The boundary of the region is defined as straight line in
  this paper for convenience. It is possible to choose the boundary
  as a curve with imaginary part extending to infinity in magnitude
  and bounded real part (i.e., a boundary dividing the complex plane
  vertically). The only necessary restriction on boundary is to
  separate the complex plane into two regions and not extending to
  infinity in magnitude of real part where the quasi-polynomial has the infinitely many zeros.
  The boundary can be vertical any curve, any closed-curve (ellipse, circle).
  \item Lemma \ref{lemma:zeroregion} is essential to calculate the quasi-polynomial zeros
  crossing the boundary. It is guaranteed that when $h$ is small,
  all the zeros of quasi-polynomial are outside the region except
  the zeros of $a(s)+b(s)$ which are initially traced.
  \item There are many methods in the literature to find the solution
  of the complex-valued differential equation (\ref{eq:contODE})
  such as Runge-Kutta, Bulirsch-Stoer method, Adams methods
  \cite{Hoffman2001}. In this paper, variable-step size $4^\textrm{th}$-order Runge-Kutta method is
  used. The algorithm is fast, easy to implement and robust due to
  local convergence of the trajectory of the quasi-polynomial zero
  shown in Theorem \ref{theorem:trajODE}.
  \item Instead of (\ref{eq:contODE}) which uses well-known Newton
  Algorithm in the literature for convergence to zero, there are other
  methods in the numerical analysis. Depending on the application
  and the \mbox{quasi-polynomial (\ref{eq:qpoly})}, different algorithms may be used \cite{Hoffman2001}.
  \item In our implementation, the delay step is adaptive in $4^\textrm{th}$-order Runge-Kutta method.
  If the next delay step is determined by
  \bd
  \Delta h_{i+1}=\Delta h_i \sqrt{\frac{\epsilon_{tz}}{|f(s_i^{h_i})|}}
  \ed where $s_i^{h_i}$ is the numerical solution of (\ref{eq:contODE}) at delay step $h_i$.
  Note that when the trajectory
  deviates from the tolerance for trajectory of the quasi-polynomial zero,
  the next step size is decreased with square root of the deviation due to (\ref{eq:Otube}).
  When the deviation of the trajectory of quasi-polynomial zero is smaller than
  tolerance, the step size is increased to reduce the computational
  effort.
  \item There are numerical methods in the literature to calculate rightmost
  zeros of the quasi-polynomials approximately in the literature.
  Our method has two advantages:
  \begin{itemize}
    \item It can be used to calculate zeros of the quasi-polynomial in any
    user-defined region.
    \item In addition to zeros at time-delay $h$, we have all the
    location of the zeros of the quasi-polynomial in the region from
    $0$ to $h$ which is convenient when the delay effects on
    stability is analyzed.
  \end{itemize}
\end{enumerate}

\section{Example} \label{sec:example}
Consider the following time-delay system from \cite{GuBook2003}:
\bd
\dot{x}(t)=\left(\begin{array}{cc}
                   0 & 1 \\
                   -1 & -1
                 \end{array}
 \right)x(t)+\left(\begin{array}{cc}
                   0 & 0 \\
                   0 & -1
                 \end{array}
 \right)x(t-\tau),\;\tau>0.
\ed By analytical methods, it is shown that the system is stable
$\tau\in[0,\pi)$ \cite{GuBook2003}. When $\tau=\pi$, the
characteristic equation of the system has imaginary zeros, $s=\pm
j$. We apply the method to find all the characteristic roots of the
system for $\tau\in[0,\pi)$ in the region $\textrm{Re(s)}\geq-1$.

The characteristic quasi-polynomial can be found as
\bd
f(s,\tau)=(s^2+s+1)+se^{-\tau s}.
\ed Following the algorithm in Section \ref{sec:themethod},
\begin{enumerate}
\item[1)] Define the region as $\textrm{Re(s)}\geq-1$ where $\sigma_0=-1$,
\item[2)] Error tolerance for trajectory is $\epsilon_{tz}=10^{-3}$,
\item[3-4)] Zero-crossing frequencies $\{\omega_i\}_{i=1}^{12}$ and the corresponding
time-delays $\{h(\omega_i)\}_{i=1}^{12}$ are given in \mbox{Table \ref{table:zerocrossfreq}},
\item[5-6)] The number of the zeros changes in the region when
$\mathbf{h}=\{0\}\cup\{h(\omega_i)\}_{i=0}^{12}$. The initial zeros in the region
are $\mathbf{z}=\{-1.00, -1.00\}$ when $h=0$,
\item[7-8)] All the zeros of quasi-polynomial in the region $\textrm{Re(s)}>-1$ at $h=\pi$ are given
in \mbox{Table \ref{table:therootsintheregion}}.
\end{enumerate}

\begin{table}[t]
\begin{minipage}[b]{0.45\linewidth}
\begin{center}
\caption{Zero-crossing frequencies and the corresponding delays} \label{table:zerocrossfreq}
\begin{tabular}{cc}
   & \\
  \hline
   & \\
  $\omega_i$ & $h(\omega_i)$ \\
   & \\
 \rowcolor[gray]{.9} $2.28$ & $0.65$ \\
 $5.00$ & $1.57$  \\
 \rowcolor[gray]{.9} $7.22$ & $1.96$  \\
 $9.23$ & $2.21$\\
 \rowcolor[gray]{.9} $11.12$ & $2.40$ \\
 $12.92$  &  $2.55$ \\
 \rowcolor[gray]{.9} $14.65$  & $2.68$  \\
 $16.33$  &  $2.79$  \\
 \rowcolor[gray]{.9} $17.97$  & $2.88$  \\
 $19.56$  &  $2.97$  \\
 \rowcolor[gray]{.9} $21.13$  & $3.05$  \\
 $22.66$ &  $3.12$  \\
  \hline
\end{tabular}
\end{center}
\end{minipage}
\hspace{0.5cm}
\begin{minipage}[b]{0.45\linewidth}
\begin{center}
\caption{The zeros of quasi-polynomial in the region when $h=\pi$} \label{table:therootsintheregion}
\begin{tabular}{cc}
   & \\
  \hline
   & \\
 $\pm j$ \\
 \rowcolor[gray]{.9} $-0.30, \pm1.00j$ \\
 $-0.27\pm2.60j$  \\
 \rowcolor[gray]{.9} $-0.47\pm4.54j$  \\
 $-0.59\pm6.52j$ \\
 \rowcolor[gray]{.9} $-0.68\pm8.51j$ \\
 $-0.75\pm10.51j$ \\
 \rowcolor[gray]{.9} $-0.80\pm12.51j$ \\
 $-0.85\pm14.50j$  \\
 \rowcolor[gray]{.9} $-0.89\pm16.50j$  \\
 $-0.93\pm18.50j$ \\
 \rowcolor[gray]{.9} $-0.96\pm20.50j$ \\
 $-0.99\pm22.50j$ \\
  \hline
  \vspace{0.05cm}
\end{tabular}
\end{center}
\end{minipage}
\end{table}

\begin{figure}[b]
\begin{minipage}[t]{4cm}
\begin{center}
\includegraphics[width=4cm]{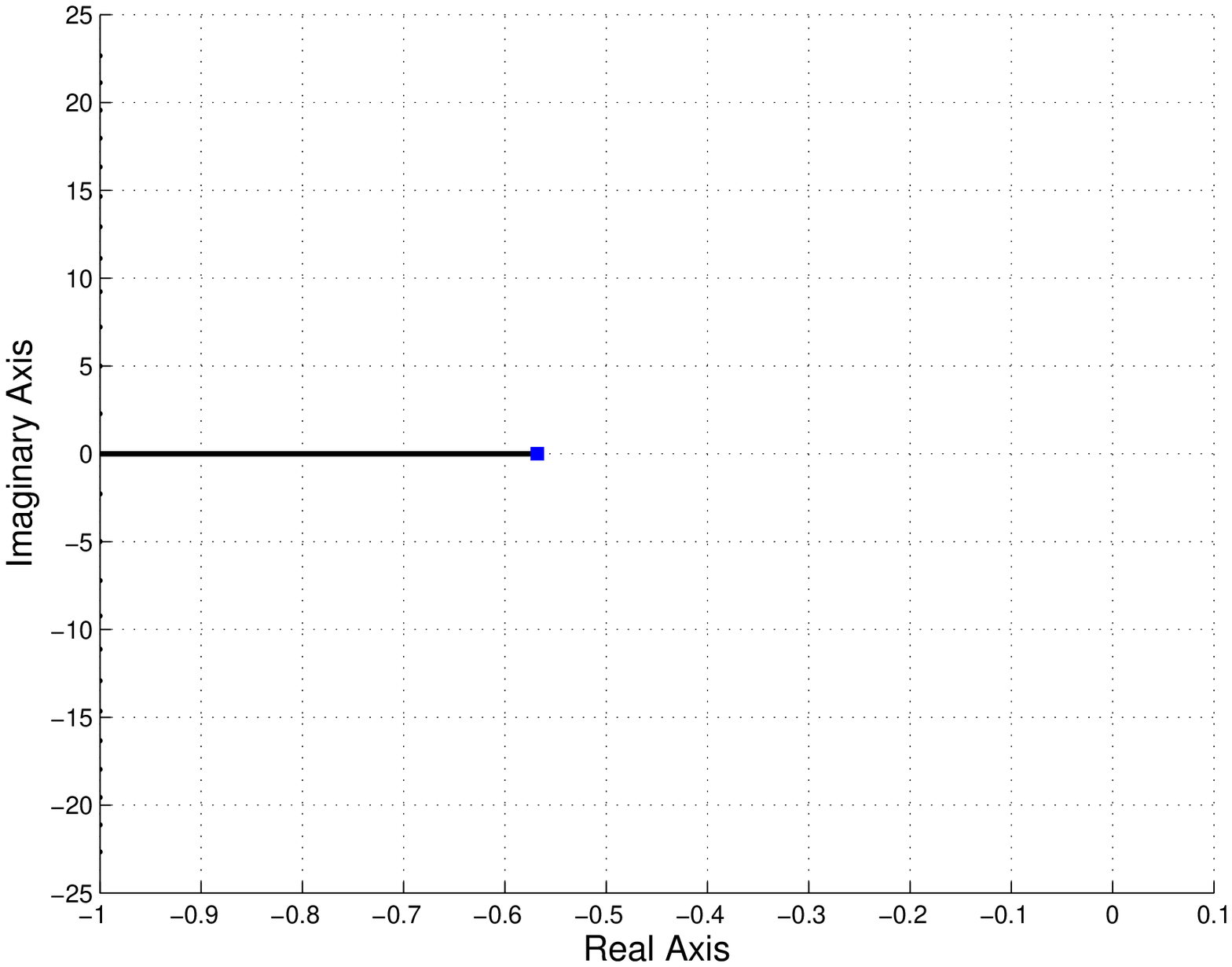}
$\tau\in[0,0.5]$
\end{center}
\end{minipage}
\hfill
\begin{minipage}[t]{4cm}
\begin{center}
\includegraphics[width=4cm]{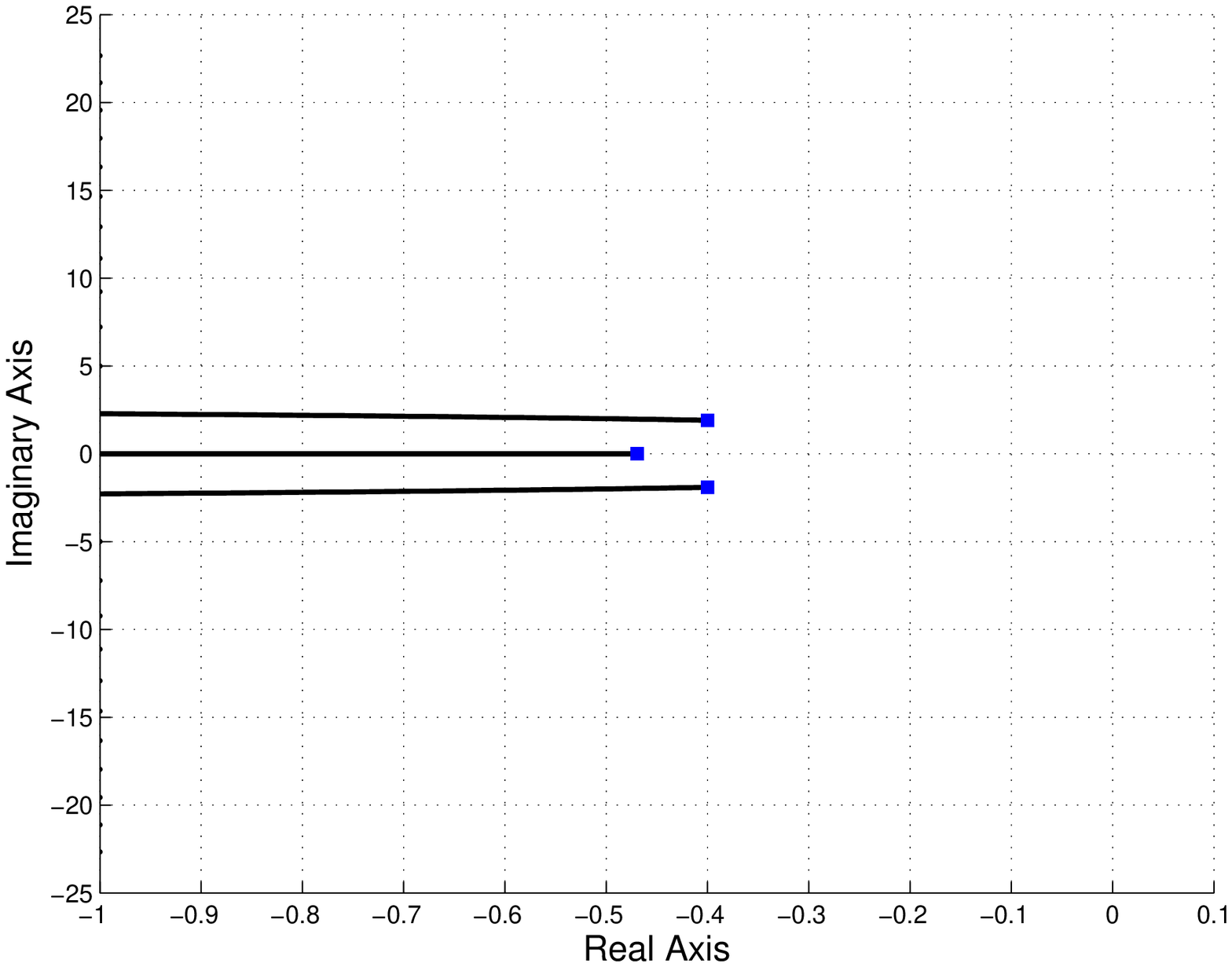}
$\tau\in[0,1]$
\end{center}
\end{minipage}
\begin{minipage}[t]{4cm}
\end{minipage}
\begin{minipage}[t]{4cm}
\begin{center}
\includegraphics[width=4cm]{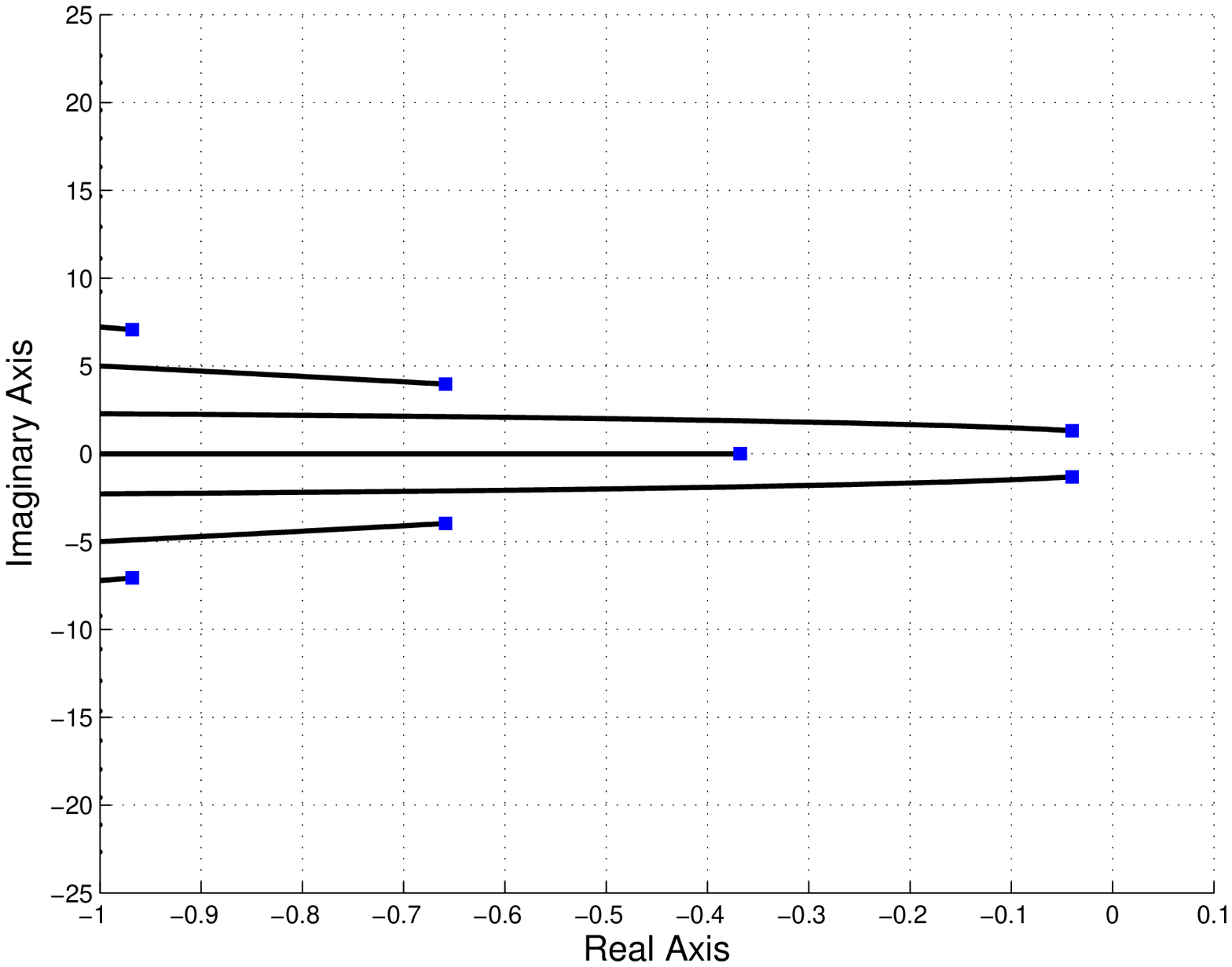}
$\tau\in[0,2]$
\end{center}
\end{minipage}
\hfill
\begin{minipage}[t]{4cm}
\begin{center}
\includegraphics[width=4cm]{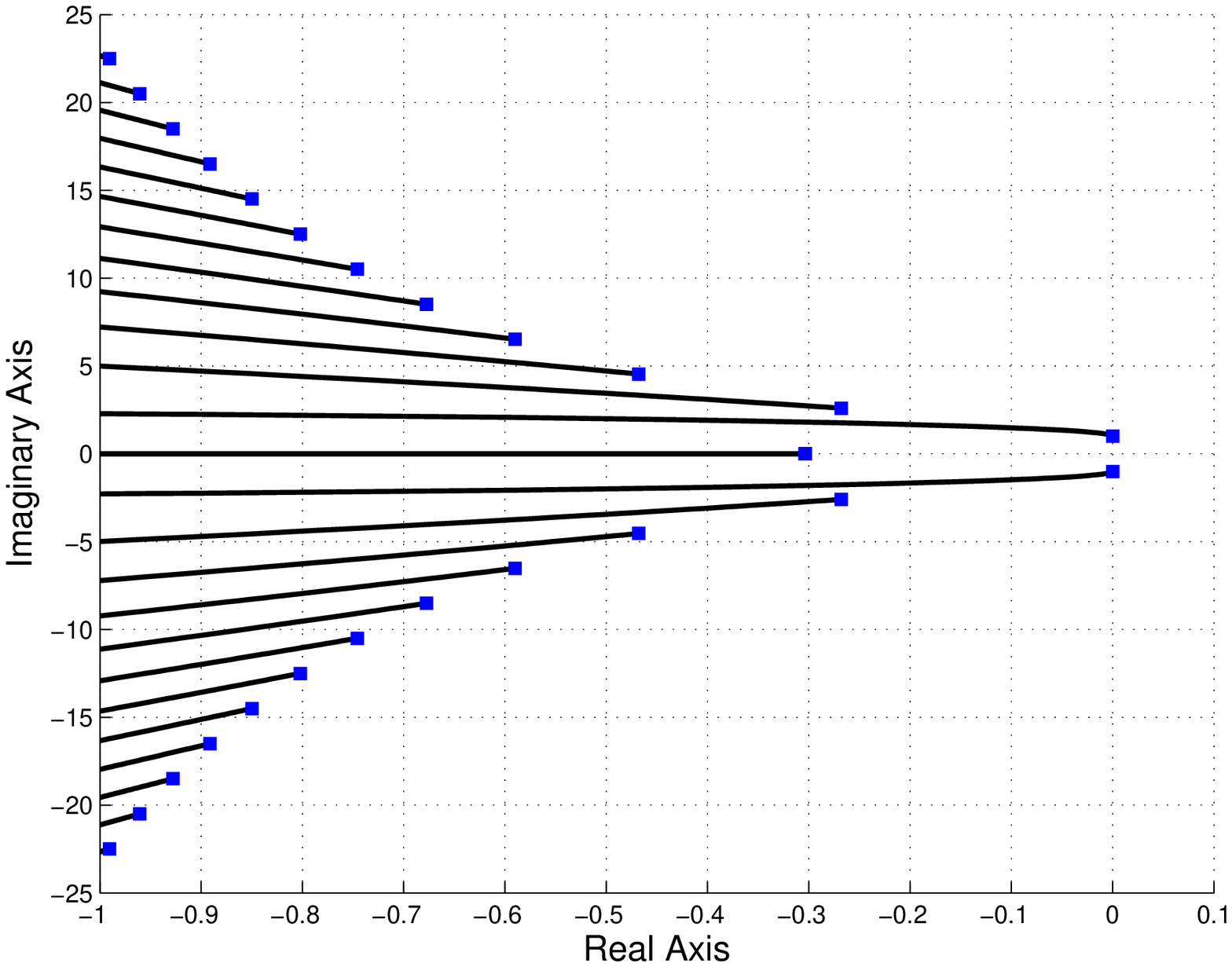}
$\tau\in[0,\pi]$
\end{center}
\end{minipage}
\caption{\label{fig:extrajofzeros} Zeros of characteristic quasi-polynomial for different delay intervals}
\end{figure}

The trajectories of the zeros of the quasi-polynomial are shown in
Figure \ref{fig:extrajofzeros}. The system is stable for
$\tau\in[0,\pi)$ as expected. When $\tau=\pi$, the system has
imaginary roots $s=\pm j$. The system has characteristic zeros close
to imaginary axis after $\tau>2.5$. Therefore, stability condition
$\tau<\pi$ is conservative. The largest error in the trajectories is
$7.6\;10^{-4}<\epsilon_{tz}$ as specified.

\section{Concluding Remarks} \label{sec:concl}
We present a novel method to find the trajectories of all the zeros of
the quasi-polynomial with single delay in predefined
right half-plane. The method is an important tool for stability
analysis of time-delay systems and effective for high-order systems.

Applications of the method to the zeros of the neutral
type quasi-polynomials and multiple delay case are the future research directions and the results will be
published elsewhere.

\section{Acknowledgments}
The author would like to thank M.~L.~Overton for his support.

\end{document}